\documentclass[10pt,reqno]{amsart}

\usepackage{amsmath}
\usepackage{amsfonts}
\usepackage{amssymb}
\usepackage{amsthm}
\usepackage{graphicx}
\usepackage{epstopdf}
\usepackage{enumitem}
\usepackage{geometry}
\usepackage{hyperref}
\usepackage{tikz}

\newtheorem{teo}{Theorem}
\newtheorem{defin}[teo]{Definition}

\newtheorem{rem}[teo]{Remark}

\newtheorem{lem}[teo]{Lemma}

\newcommand{\beq}{\begin{equation}}
\newcommand{\eeq}{\end{equation}}

\begin{document}

\title{SOLUTION OF WALD'S GAME USING LOADINGS AND ALLOWED STRATEGIES}

\author{Valerio Capraro}
\address{University of Neuchatel, Switzerland}
\thanks{Supported by Swiss SNF Sinergia project CRSI22-130435}
\email{valerio.capraro@unine.ch}



\date{}

\maketitle

\begin{abstract}
We propose a new interpretation of the strange phenomena that some
authors have observed about the Wald game. This interpretation is
possible thanks to the new language of \emph{loadings} that Morrison
and the author have introduced in a previous work. Using the theory
of loadings and allowed strategies, we are also able to prove that
Wald's game admits a \emph{natural} solution and, as one can expect,
the game turns out to be fair for this solution. As a technical
tool, we introduce the notion of \emph{embedding a game into another
game} that could be of interest from a theoretical point of view.
\emph{En passant} we find a very easy example of a game which is
loadable in infinitely many different ways.
\end{abstract}

\section{Introduction}

Consider Wald's game \emph{pick the bigger integer}: the set of pure
strategies is the set of non-negative integers and the payoff
function of player 1 (which is the negative of the payoff function
of players 2) is

$$
f(s,t)=\left\{
         \begin{array}{ll}
           1, & \hbox{if $s>t$} \\
           0, & \hbox{if $s=t$} \\
           -1, & \hbox{if $s<t$}
         \end{array}
       \right.
$$

Wald\cite{Wa} proved that this game has no value if just countably
additive strategies are allowed. On the other hand, in \cite{He-Su}
it has been shown that this game has a value if one allows finitely
additive probability measures as strategies, but the value depends
on the order of integration in such a way that the \emph{internal}
player has an advantage\footnote{See \cite{Sc-Se} for more general
results and relations with other phenomena, as de Finetti's
non-conglomerability.}. So a game which is naturally symmetric turns
out to be asymmetric.

This note has two goals: first we want to show that this apparently
strange situation is perfectly explained and natural looking at the
game from the point of view of the theory of loadings and allowed
strategies; second we want to propose a natural solution of the game
using again these new notions. It is interesting the fact that
Wald's game turns out to be fair for such a \emph{natural} solution,
that is exactly what one would expect.

In the next section we recall very briefly some of the definitions
given in \cite{Ca-Mo} and we introduce the notion of embedding of
games that could be of interest from a theoretical point of view
(see Remark \ref{approximation}). The final section is devoted to
the results that motivate this note.

\section{Basic notions}

Let $\mathcal G_1$ and $\mathcal G_2$ be two $n$-player games and
let $P_j^i$ (resp. $S_j^i$), for $i\in\{1,...,n\}$ and $j=1,2$, be
the set of pure (resp. mixed) strategies of the $i$-th player of the
$j$-th game. Let us denote by $\pi_j$ the payoff function of the
game $\mathcal G_j$. Given a map $f:P_1^i\rightarrow P_2^i$ and
$\mu\in S_1^i$, we can consider the push-forward strategy $f_*\mu\in
S_2^i$.

\begin{defin}\label{embedding}
We say that
\begin{enumerate}
\item $\mathcal G_1$ embeds into $\mathcal G_2$ if there exist $n$ injective maps
$f_i:P_1^i\rightarrow P_2^i$ such that for all $(s_1^1,...s_1^n)\in
S_1^1\times...\times S_n^1$ one has
$$
\pi_1(s_1^1,...s_1^n)=\pi_2(f_{1^*}s_1^1,...f_{n^*}s_1^n))
$$
\item $G_1$ is isomorphic to $G_2$ if each $f_i$
is bijective.
\end{enumerate}
\end{defin}

\begin{rem}{\bf (Approximation of games)}\label{approximation}
The notion of embedding of games leads very naturally to some notion
of approximation of a game by taking an increasing family of
subgames that converges in some reasonable sense. This could be a
useful tool to study complicated games using easier ones. In order
to develop a good theory, one would need some theorem of convergence
and this is certainly the purpose of future research. Here we are
just interested in giving a first application of this notion.
\end{rem}

Recall that a \emph{semigroup} $S$ is a set equipped with an
associative binary operation $S\times S\rightarrow S$. Given $x,y\in
S$, the result of the operation is denoted by $xy$.
\begin{defin}\label{operation}
Let $S$ be a semigroup and let $W$ be a subset of $S$. Fix a
function $h:S\rightarrow[-1,1]$ The Operation Game associated to
$S$, $W$ and $h$, denoted by $\mathcal G(S,W,h)$, is the two-person
zero-sum game with $S$ the set of pure strategies for both players:
player 1 chooses $x\in S$ and player 2 chooses $y\in S$, then player
1 wins if $xy$ is in $W$. The payoff to player 1, which is the
negative of the payoff to player 2, is $h(xy)$.\footnote{The
definition in \cite{Ca-Mo} was actually little different, but the
reader can easily prove that they are basically equivalent. Indeed,
in \cite{Ca-Mo}, we use $h=\chi_W$ and we can pass to function
taking values in $[-1,1]$ without modifying the game with the
(usual) affine transformation
$[0,1]\ni\lambda\rightarrow2\lambda-1\in[-1,1]$. Finally each proof
in \cite{Ca-Mo} uses just the invariance of the measures and so it
could be applied to a a general function $h$, instead of the
particular one $\chi_W$.}
\end{defin}

\begin{rem}\label{wald}
We can also consider the more general case in which the set of pure
strategies of player $i$ is just a subset $S_i$ of $S$. In this case
the Operation Game will be denoted by $\mathcal G(S_1,S_2,S,W,h)$.
In particular we are interested in the operation game $\mathcal
G(\mathbb N,-\mathbb N,\mathbb Z,\mathbb N,\chi_{\mathbb
N}-\chi_{\mathbb Z\setminus\mathbb N})$, so that the payoff to
player 1 is
$$
f(x,y)=\left\{
         \begin{array}{ll}
           1, & \hbox{if $x+y\in\mathbb N$} \\
           0, & \hbox{if $x+y=0$} \\
           -1, & \hbox{if $x+y\in-\mathbb N$}
         \end{array}
       \right.
$$
\end{rem}
In \cite{Ca-Mo} it is shown that loadings play a fundamental rule to
play in a coherent way and to solve the Operation Game. Indeed in
general it is typical the situation that the game, which is
intrinsically symmetric, loses its symmetry. Here is a sketch of the
construction. First we recall the following

\begin{defin}
A \textbf{loading} on $S$ is any finitely additive probability
measure on $S$ which is left and right invariant with respect to the
operation in $S$. A loading is denoted by $\ell$.
\end{defin}

The loading is fixed a priori and it is \emph{part of the rules of
the game}, in the sense that it induces a set of allowed strategies
as follows

\begin{defin}\label{allowed}
The set of allowed strategies $\mathcal A_\ell$ is any maximal set
of finitely additive probability measures on $S$ containing the
loading $\ell$ and keeping the symmetry of the game, i.e.
$$
\int_S\int_S\chi_W(xy)dp(x)dq(y)=\int_S\int_S\chi_W(xy)dq(y)dp(x)
$$
for all $p,q\in\mathcal A_\ell$ (where $\chi$ stands for the
characteristic function).
\end{defin}

In \cite{Ca-Mo} the authors have shown that allowing to play just
the strategies in $\mathcal A_l$, then the Operation Game in the
sense of Definition \ref{operation} has the value $\ell(W)$ and
there are at least one optimal strategies for both the players,
which is indeed $\ell$. Moreover, it is shown that there are
operation games which are loadable in infinitely many different
ways. This basically means that such games are dramatically not well
defined and it is really necessary to fix the loading a priori in
order to play in a coherent way. We are going to show that this is
exactly what happens in the case of Wald's game. By the way we
recall that in many cases (finite games, compact games, other lucky
cases) the choice of a loading is not required, because there is
basically a unique way to load the game and so the game is
intrinsically well defined. Unfortunately, this is not the case also
for very natural groups and semigroups, such as $(\mathbb Z,+)$ and
$(\mathbb N,\cdot)$, as is also shown in Lemma \ref{lemma}.

\section{Main results}

The following lemma is quite unexpected but probably well-known to
the experts. In our case it is necessary for the proof of the main
result and also it provides a very easy example of a game which is
loadable in infinitely many different ways. We recall that at the
end of \cite{Ca-Mo}, it is proposed an example of a game which is
loadable in infinitely many different ways, but this example is
quite complicated and unnatural from the point of view of a
\emph{real} game, namely a game that can be proposed to human
beings.

\begin{lem}\label{lemma}
For any real number $s\in[0,1]$, there is a finitely additive
translation invariant probability measures $m$ on $\mathbb Z$ such
that $m(\mathbb N)=s$. In particular there are uncountably many
different ways to load the Operation Game $\mathcal G(\mathbb
Z,\mathbb N)$.
\end{lem}

\begin{proof}
Let $k$ be a positive integer. If $A$ is a subset of $\mathbb Z$, we
define its $k$-density as follows
$$
d_k(A)=\lim_{n\rightarrow\infty}\frac{|A\cap[-kn,n]|}{(k+1)n+1}
$$
where, as usual, the notation $[n,m]$ stands for the set of integers
$x$ such that $n\leq x\leq m$. Clearly there are many subsets $A$
for which $d_k$ does not exist, so let $D_k=\{\chi_A, A \text{ s.t.
} d_k \text{ exists}\}$ and let $X_k\subseteq L^\infty(\mathbb Z)$
be the linear span of $D_k$ inside $L^\infty(\mathbb Z)$. Clearly
$d_k$ extends to a linear and bounded functional on $X_k$ which is
dominated by the $L^\infty$-norm. Hence we can apply the Hahn-Banach
extension theorem to get a positive linear functional
$\phi_k:L^\infty(\mathbb Z)\rightarrow\mathbb R$. Now let $\mathcal
X_k\subseteq (L^\infty(\mathbb Z))^*$ be the set of all such
extensions $\phi_k$. It is a convex and compact space with respect
to the weak*-topology. Convexity is indeed trivial and compactness
follows from the fact that $\mathcal X_k$ is weak* closed in the
unit ball of $(L^\infty(\mathbb Z))^*$, which is weak* compact by
the Banach-Alouglu theorem. Now $\mathbb Z$ acts on $\mathcal X_k$
via translations as a commutative family of operators and then
Markov-Kakutani fixed point theorem applies. Such a fixed point, say
$f_k$, is an element in $L^\infty(\mathbb Z)^*$ which is fixed by
every translation and such that $f_k(1)=1$. It follows that setting
$m_k(A)=f(\chi_A)$ we get a translation invariant finitely additive
probability measure on $\mathbb Z$. Now since $f_k$ extends $d_k$,
we have $m_k(\mathbb N)=d_k(\mathbb N)=\frac{1}{(k+1)}$. When $k$
goes to infinity, this proves that there exist finitely additive
translation invariant probability measures taking over $\mathbb N$
values arbitrarily close to $0$. Now repeating the argument with
$$
d_k(A)=\lim_{n\rightarrow\infty}\frac{|A\cap[-n,kn]|}{(k+1)n+1}
$$
we also prove that there exist finitely additive translation
invariant probability measures which take over $\mathbb N$ values
arbitrarily close to $1$. Now we know - and this is basically due to
Chou \cite{Ch} - that the set of values which are taken by some
finitely additive translation invariant probability measure on a
fixed set $W$ is convex and closed, so, in our case, it has to be
the whole interval $[0,1]$.
\end{proof}

\begin{teo}\label{uno}
Wald's game is equivalent to an Operation Game which is loadable in
infinitely many different ways.
\end{teo}

\begin{proof}
The idea is easy: consider the Operation Game $\mathcal G(\mathbb
Z,\mathbb N,h)$, with $h(x)=\chi_{\mathbb N}(x)-\chi_{\mathbb
Z-\mathbb N}(x)$. We want to embed Wald's game into this game. More
precisely we are going to show that Wald's game is equivalent to the
game $\mathcal G(\mathbb N,-\mathbb N,\mathbb Z, \mathbb
N,\chi_{\mathbb N}-\chi_{\mathbb Z\setminus\mathbb N})$ of Remark
\ref{wald} and so in particular, the set of outcomes\footnote{The
set of outcomes in this context is the set of $xy$, when $x\in S_1$
and $y\in S_2$.} is still $\mathbb Z$. This is important in order to
apply Lemma \ref{lemma}, since it is certainly false in general that
a subgame of a game which is loadable in infinitely many different
ways is still loadable in infinitely many different ways, but a
sufficient condition to pass this property to subgames is clearly
that the set of outcomes $S$ remains the same. Now that we have the
idea, the proof is straightforward: with the notation of Definition
\ref{embedding}, define $f_1(x)=x$ and $f_2(y)=-y$. So Player $1$,
after transforming the game, wins if and only if
$f_1(x)+f_2(y)\in\mathbb N$ that happens if and only if
$\max\{x,y\}=x$. Hence Player 1 wins in the Operation Game $\mathcal
G(\mathbb N,-\mathbb N, \mathbb Z, \mathbb N,\chi_{\mathbb
N}-\chi_{\mathbb Z\setminus\mathbb N})$ if and only if it wins in
the Wald game and the payoff function is clearly preserved. This
proves that these two games are isomorphic and in particular Wald's
game is loadable in infinitely many different ways.
\end{proof}

Now we are ready to propose a solution for the Wald game. There is
indeed a very natural approach to solve the game $\mathcal G(\mathbb
N,-\mathbb N,\mathbb Z,\mathbb N,\chi_{\mathbb N}-\chi_{\mathbb
Z\setminus\mathbb N})$, given by the fact that the set of pure
strategies of the first player is the opposite of the set of pure
strategies of the second player. This observation leads to the fact
that a loading that can be \emph{accepted for playing by both the
players} is any loading verifying $\ell(A)=\ell(-A)$, for all
$A\subseteq\mathbb N$. It indeed reflects the point of view of both
the players. It is clear that any loading verifying such equality is
such that $\ell(\mathbb N)=\frac{1}{2}$. Hence

\begin{teo}
Wald's game can be solved in a \emph{natural} way and the value of
the game for this solution is $0$, i.e. the game is fair.
\end{teo}

\begin{proof}
We have already observed that a natural way to solve the game
consists in fixing a loading $\ell$ verifying $\ell(A)=\ell(-A)$,
for $A\subseteq\mathbb N$, and then playing the allowed strategies
induced by $\ell$. The point is that we cannot use directly
Definition \ref{allowed} to construct the set of allowed strategies,
since the mixed strategies of the players are not finitely additive
probability measures on $\mathbb Z$. Indeed the strategies of Player
1 (resp. Player 2) are measure on $\mathbb N$ (resp. $-\mathbb N$).
But we can take inspiration from Definition \ref{allowed} and make
the following construction. First of all, observe that given a
measure $\mu$ on $\mathbb N$, we can construct a probability measure
$\hat\mu$ on $\mathbb Z$ \emph{by reflecting} $\mu$ in the following
way: given $A\subseteq\mathbb Z$, we define
$$
\hat\mu(A)=\mu(A\cap\mathbb N)+\mu(-((A-1)\cap(-\mathbb N)))
$$
where $A-1=\{a-1,a\in A\}$. Analogously we can construct a
probability measure $\hat\mu$ on $\mathbb Z$ by reflecting a
probability measure $\mu$ on $-\mathbb N$. Let $\mathcal P(X)$ be
the set of probability measures on a set $X$, we define the set of
allowed strategies for $\mathcal G(\mathbb N,-\mathbb N,\mathbb
Z,\mathbb N,\chi_{\mathbb N}-\chi_{\mathbb Z\setminus\mathbb N})$
induced by $\ell$ to be any subset $\mathcal A_\ell\subseteq
\mathcal P(\mathbb N)\times \mathcal P(-\mathbb N)$ which is maximal
with respect to the following two properties:
\begin{enumerate}
\item
$$
\int_{\mathbb N}\int_{-\mathbb N}(\chi_{\mathbb
N}(x+y)-\chi_{\mathbb Z\setminus\mathbb
N}(x+y))dp(x)dq(y)=\int_{-\mathbb N}\int_{\mathbb N}(\chi_{\mathbb
N}(x+y)-\chi_{\mathbb Z\setminus\mathbb N}(x+y))dq(y)dp(x)
$$
for all $(p,q)\in\mathcal A_\ell$.
\item $\ell\in\{\hat\mu,\mu \text{ allowed strategy}\}$
\end{enumerate}
Observe that this maximal set exists and contains the pair $(\check
\ell_1,\check\ell_2)$, where $\check\ell_1\in\mathcal P(\mathbb N)$
and $\check\ell_2\in\mathcal P(-\mathbb N)$ are defined by setting
$$
\check\ell_1(A)=\ell(A)+\ell(-A)
$$
for all $A\subseteq\mathbb N$ ($\check\ell_2\in\mathcal P(-\mathbb
N)$ is defined in an analogue way). It is now easy to check, being
the details basically the same as the main theorem in \cite{Ca-Mo},
that $(\check\ell_1,\check\ell_2)$ is a profile of optimal
strategies and that the value of the game is indeed $2\ell(\mathbb
N)-1=2\cdot\frac{1}{2}-1=0$.
\end{proof}


\begin{thebibliography}{9999}

\bibitem[Ca-Mo]{Ca-Mo} V. Capraro and K.Morrison, \emph{Existence of optimal strategies for the operation game on amenable semigroup},
Preprint (2011)

\bibitem[Ch]{Ch} C. Chou, \emph{On topologically invariant means on a
locally compact group}, Trans. Amer. Math. Soc. 151 (1970) 443--456.

\bibitem[He-Su]{He-Su} D. Heath and W. Sudderth, \emph{On a theorem of de Finetti, oddsmaking, and game theory},
Ann. of Math. Stat. 43 (1972) 2072--2077.

\bibitem[Sc-Se]{Sc-Se} M. J. Schervish and T. Seidenfeld, \emph{A fair minimax theorem for two-person (zero-sum) games involving finitely additive strategies},
in \emph{Rethinking the Foundations of Statistics}, ed: J. B.
Kadane, M. J. Schervish, and T. Seidenfeld, Cambridge: Cambridge
Univ. Press (1999) 267--291.

\bibitem[Mo]{Mo} K. E. Morrison, \emph{The multiplication game},
Math. Mag. 83 (2010) 100--110.

\bibitem[Wa]{Wa} A. Wald, \emph{Statistical decision functions}, New
York: Wiley (1950).

\end{thebibliography}
\end{document}